\documentclass[conference]{vldb}
\pdfoutput=1
\usepackage{times}
\usepackage{graphicx}
\usepackage{epstopdf}

\usepackage{subfigure}
\usepackage[linesnumbered,ruled,vlined]{algorithm2e}
\begin{document}
\newtheorem{definition}{Definition}
\newtheorem{example}{Example}
\newtheorem{theorem}{Theorem}
\newtheorem{lemma}{Lemma}

%
%
%
%

\title{ Subjective Knowledge Acquisition and Enrichment\\ Powered By Crowdsourcing }

\author{Rui Meng~~~~Hao Xin~~~~Lei Chen~~~~Yangqiu Song\vspace{3pt} \\
	\affaddr{Department of Computer Science and Engineering, HKUST, Hong Kong SAR, China} \\
	\affaddr{\{rmeng,hxinaa,leichen,yqsong\}@cse.ust.hk}
}
%

\maketitle
\begin{abstract}
	
	Knowledge bases (KBs) have attracted increasing attention
	due to its great success in various areas, such as Web and mobile search.
	Existing KBs are restricted to objective factual 
	knowledge, such as \textsc{city population} or \textsc{fruit shape}, whereas,
	subjective knowledge, such as \textsc{big city}, which is commonly mentioned
	in Web and mobile queries, has been neglected. Subjective knowledge
	differs from objective knowledge in that it has no documented or
	observed ground truth. Instead, the truth relies on people's dominant
	opinion. Thus, we can use the crowdsourcing technique to
	get opinion from the crowd. In our work, we propose a system,
	called \textit{\underline{c}r\underline{o}wdsourced \underline{s}ubjective \underline{k}nowledge \underline{a}cquisition (CoSKA)},  
	for subjective knowledge acquisition powered by crowdsourcing
	and existing KBs. The acquired knowledge can be used to enrich
	 existing KBs in the subjective dimension which bridges the
	gap between existing objective knowledge and subjective queries.
	The main challenge of \textit{CoSKA} is the conflict between large scale
	knowledge facts and limited crowdsourcing resource. To address this
	challenge, in this work, we define knowledge inference rules and
	then select the seed knowledge judiciously for crowdsourcing to
	maximize the inference power under the resource constraint. Our
	experimental results on real knowledge base and crowdsourcing
	platform verify the effectiveness of \textit{CoSKA} system.

\end{abstract}

%
%
%
%
%
%



\section{Introduction}
\label{sec:introduction}
\textbf{Motivation}. In recent years, knowledge bases (KBs) have become increasingly popular and large-scale KBs have been constructed, such as Freebase~\cite{freebase}, DBpedia~\cite{dbpedia}, YAGO~\cite{YAGO}, KnowItAll~\cite{knowitall}, etc. The KBs encode information and knowledge of the real world in a structured, machine-understandable way which can empower various kinds of applications, especially Web and mobile search. Despite of containing millions of knowledge facts on large amount of entities and relations, the knowledge encoded by these KBs is limited in objective dimension. In other words, existing KBs have so far focused on encoding objective knowledge facts, which are factual and
observable, such as \textsc{fruit shape}, \textsc{movie director} and so forth. In contrast, many real world queries are subjective, e.g., around 20\% of product-related queries are labeled as being ``subjective'' by workers~\cite{subjective_statistics}, 63\% of location-based queries in mobile search are asking for subjective opinions~\cite{location_based_opinion}, and need the corresponding subjective knowledge as the query answers.  For example, there might exist such queries, ``\textit{popular American singers}'' or ``\textit{beautiful cities in Europe}'', we refer the knowledge concerning \textit{popular singers} and \textit{beautiful cities} as the \textit{subjective knowledge}. More specifically, subjective knowledge refers to the \textit{dominant opinion} about whether a particular subjective property applies to entities of a particular type~\cite{miningsubjective}. For instance, given a pair consists of a subjective property\footnote{Typically expressed as an adjective} and a type from a KB (subjective property-type pair, ST pair), e.g., \textsc{popular} and \textsc{singer}, we can find a list of instances of the type \textsc{singer} from the KB, e.g., \textsc{Elvis Presley}, where the dominant opinion of \textit{``whether Elvis Presley is a popular singer''} is a piece of subjective knowledge. 
As this kind of information is missing in existing KBs, queries concerning such information cannot be satisfied. Fortunately, crowdsourcing, which has been recently proved to be successful for various human intrinsic tasks such as entity resolution~\cite{CrowdER}, knowledge extraction~\cite{CrowdsourcedKnowledge}, translation~\cite{Crowdsourcingtranslation}, etc., provides a natural and reliable way of obtaining the subjective knowledge by collecting opinions from  workers.

  
Many works have been done to perform KB enrichment, completion and population~\cite{AMIE}~\cite{kbcompletion}~\cite{patternknowledgeenrich}~\cite{KBP_Successful}~\cite{KBP_TacklingChallenge}, 
but none of these works focus on the subjective dimension. For subjective knowledge acquisition, the state-of-the-art approach is to use information extraction techniques to mine the text of Web contents~\cite{miningsubjective}. However, 
it only relies on machine-based technique and online Web data, and does not consider to incorporate the wisdom of the crowd and existing KB information. Thus, the precision is far from satisfactory, i.e. \textsc{surveyor} has the precision of 77\%~\cite{miningsubjective}.
To the best of our knowledge, we are the first to leverage the collaborative knowledge from both the crowd and existing KBs to perform subjective knowledge acquisition and KB enrichment in the subjective dimension.

\textbf{Challenge}. Leveraging the power of the crowd for knowledge acquisition comes with the challenge of ``How to resolve the conflict between large scale knowledge facts and the limited crowdsourcing resource?''. 
Real world KBs are often in very large scales, e.g., YAGO has 2,747,873 entities and 292,898 types, DBpedia has 2,531,369 entities and 827 types, while each crowdsourcing operation is associated with a monetary cost and is somewhat time-consuming. Therefore, it is infeasible and costly to ask the crowd to carry the whole burden of subjective knowledge acquisition task. 
In our system \textit{CoSKA}, we make use of the knowledge in existing KBs and the semantic relationship among subjective properties to perform knowledge inference and based on the inference power, the most beneficial questions are identified for crowdsourcing.
\begin{figure}
	\includegraphics[width=0.44\textwidth,height=0.18\textwidth]{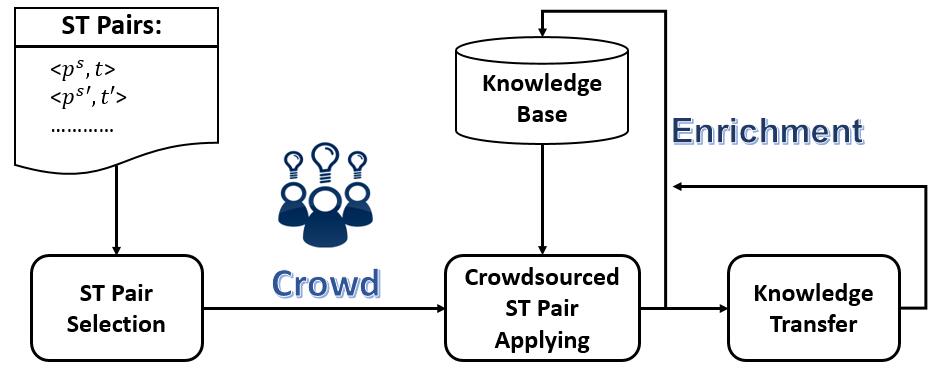}
	\caption{Framework of \textit{CoSKA}.}
	\label{fig:framework}
	\vspace{-2ex}
\end{figure}

\textbf{Framework}. The input of \textit{CoSKA} is a list of ST pairs mined from the corpus and a  KB. The output is a list of subjective knowledge facts and enriched KB. \textit{CoSKA}  consists of three stages: ST pair selection, crowdsourced ST pair applying, and knowledge inference. 
The details of the framework is shown in Figure~\ref{fig:framework}.

\textbf{1) ST Pair Selection}: Given the large amount of ST pairs, we need to identify the benefit of each pair and select them judiciously for subsequent crowdsourced ST pair applying as the process needs the involvement of crowd workers. We first define some subjective knowledge inference rules. Then
the ST pair selection problem is formulated as a \textit{Maximum Knowledge Inference Problem}. We show that the problem is NP-hard and propose a diversity-aware forward greedy algorithm for ST pair selection. 

\textbf{2) Crowdsourced ST Pair Applying}: For each selected ST pair, the task of subjective knowledge acquisition is to identify the opinions that whether the subjective property can be applied to the instances of the type powered by crowdsourcing, referred as crowdsourced ST pair applying.  However, asking the crowd for every instance is still too costly as a type could contain hundreds of thousands instances. In order to improve the scalability of the knowledge acquisition task, we formulate the crowdsourced ST pair applying as a binary classification problem. The objective knowledge of instances in existing KBs is selected as the features. 
We adopt a representative sampling strategy  to sample a set of instances to ask the crowd and the classifier is trained based on the collected answers.

\textbf{3) Knowledge Inference}: After the crowdsourced ST pair applying process, we have acquired a set of subjective knowledge facts. To further improve the scalability of our system and derive more subjective knowledge facts, we perform knowledge inference based on the subjective inference rules. 

The acquired and inferred knowledge can be encoded into existing KBs to perform KB enrichment in the subjective dimension. 

In summary, the contributions of our work are as follows:
\begin{itemize}
\item We propose the problem of crowdsourced subjective knowledge acquisition and perform knowledge base enrichment in the subjective dimension, which bridges the gap between the subjective queries and existing knowledge bases encoding only objective knowledge.

\item  We describe and implement our \textit{CoSKA} system, consists of ST pair selection, crowdsourced ST pair applying and knowledge inference, for crowd-powered subjective knowledge acquisition.

\item  We define subjective knowledge inference rules among ST pairs and formulate the ST pair selection problem as a \textit{Maximum Knowledge Inference Problem}. We prove the problem is NP-hard and propose a diversity-aware forward greedy algorithm for ST pair selection.

\item  To further resolve the conflict between large scale knowledge facts and the limited crowdsourcing resource, we formulate the crowdsourced ST pair applying problem as a classification task and derive more knowledge facts based on the crowdsourced seed knowledge. 

\item We conduct extensive experiments using real large-scale knowledge base and crowdsourcing platform and verify the effectiveness of \textit{CoSKA} system. 

\end{itemize}


The rest of the paper is organized as follows. In Section~\ref{sec:definition}, we introduce  preliminaries and give the formal definitions of subjective knowledge acquisition and enrichment. 
In Section~\ref{sec:pair_selection}, we present the methodology for ST pair selection. In Section~\ref{sec:crowdsourcing_knowledge}, we describe the models for crowdsourced ST pair applying.
The crowdsourcing mechanism design is illustrated in Section~\ref{sec:crowd}. 
 Section~\ref{sec:experiment} shows the experimental results on real KBs and crowdsourcing platform. The related works are introduced in Section~\ref{sec:related}. We conclude our work in Section~\ref{sec:conclusion}.

\section{Problem Definition}
\label{sec:definition}
A knowledge base is a repository of storing entities and relations in a real world scenario. Similar with~\cite{DBLP:journals/corr/abs-1207-4525},  knowledge base is formally defined as follows.
\begin{definition}[Knowledge Base]
	\label{def:knowledgebase}
A knowledge base KB is a tuple denoted by $(E,L,R,P)$, consisting of a collection of entities $E$, literals $L$, relations $R$ holding between entities, and properties $P$ holding between entities and literals. An entity $e \in E$ can be a class or an instance.
\end{definition}

Figure~\ref{fig:knowledge} shows a toy example of a knowledge base. 
There are six entities - three classes, e.g., ``lawyer", ``politician", and ``president", and three instances, e.g., ``Obama", ``Michelle ", and ``Gorge W. Bush"; the date ``1961-8-4"  and string ``Barack Obama"  are literals; there are three relations (``type", ``married", and ``subclassOf'') and two kinds of properties (``birthDate" and ``fullName"). 

\begin{figure}
	\centering
	\includegraphics[width=0.4\textwidth]{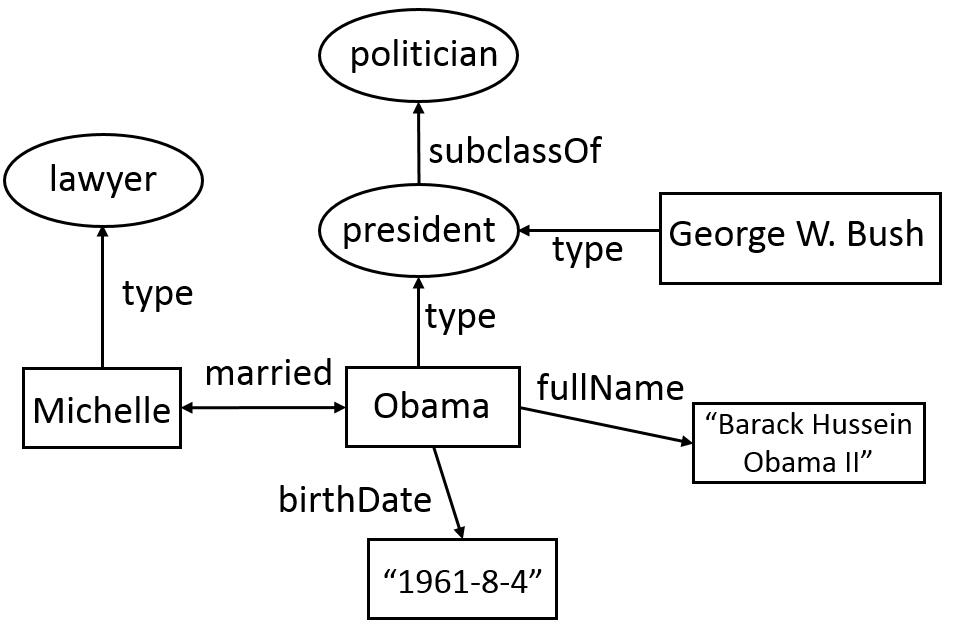}
	\caption{An example of knowledge base.}
	\vspace{-2ex}
	\label{fig:knowledge}
\end{figure} 

\begin{definition}[Objective Knowledge]
Objective knowledge is a fact of triple $<s,p^o,o>$, where $s$ is an entity in a knowledge base, $p$ is an objective property, and $o$ is either an entity or a literal. Objective knowledge recording the real world facts, which is factual and observable.
\end{definition}

As shown in the toy example of a KB of Figure~\ref{fig:knowledge}, there are eight objective knowledge facts, e.g. $<Obama, birthDate, ``\text{\textit{1961-8-4}}''>$ and $<president, subclassOf,  policitian>$.

As mentioned in Section~\ref{sec:introduction}, the subjective knowledge refers to the dominant opinion about whether a particular subjective property applies to entities of a particular type~\cite{miningsubjective}. Therefore, the combination of a subjective property and a certain type should be figured out for subsequent subjective knowledge acquisition. We define it as the \textit{ST pair} (subjective property-type pair).
\begin{definition}[ST Pair]
	An ST pair consists of a subjective property  and a type, which corresponds to a class entity in the knowledge base, i.e. $ST=(p^s,T)$. An ST pair for knowledge acquisition task indicates that the subjective property $p^s$ can be applied to the type $T$, namely, can be applied to the instances of the type $T$. 
\end{definition}

For example, an ST pair, $ST= (big, City)$ indicates that the entities of the $City$ type have the subjective property of $big$. Same for the ST pairs like $(cute, Animal)$,  $(popular, Sport)$ and etc. 

\begin{definition}[Subjective Knowledge]
A subjective knowledge fact is a triple denoted by $<s,ST,l>$, where $s$ is an entity in a knowledge base, $ST$ is an ST pair consists of a subjective property and a type, and $l$ is a label with value either be $true$ or $false$. 
\end{definition}

The subjective knowledge has no ground truth, instead it has a \textit{dominant opinion} which can be used to derive such knowledge. For example, if most people hold the opinion that ``New York is a big city'', then we can derive a new subjective knowledge fact $<New York, (big,City), true>$; otherwise, we will derive a new subjective knowledge fact $<New York, (big,City), false>$.

\begin{definition}[ST Pair Applying]
	Given an ST pair, $ST=<p^s, T>$, consists of a subjective property $p^s$ and a type $T$ , and a knowledge base $KB$, ST pair applying refers to the process of deciding whether $p^s$ can be applied to the instances of type $T$ in the $KB$, and the result is a list of subjective knowledge facts, $\mathcal{F}^s =\{F_1,F_2,\cdots, F_m\}$, where $F_i=\{e_i,ST,l\}$. 
\end{definition}

Given a list of ST pairs, subjective knowledge acquisition refers to the process of performing ST pair applying for all the input ST pairs. The derived knowledge can be encoded into the existing KB to perform KB enrichment in the subjective dimension.

\begin{definition}[Subjective Knowledge Enrichment]
\label{def:problem_definition}
Given a knowledge base $KB$ consisting objective knowledge facts $\mathcal{F}^O$, a list of 
$ST$ pairs, $\mathcal{ST}=\{{ST}_1,{ST}_2,\cdots,{ST}_m\}$, the target is to enrich the $KB$ with a list of subjective knowledge facts $\mathcal{F}^S$ by performing subjective knowledge acquisition for the ST pairs.
\end{definition}



We employ crowdsourcing for the subjective knowledge acquisition, specifically for the ST pair applying process. 
 Due to the limited crowdsourcing resource, we need to crowdsource in an efficient and productive manner. In other words, for crowdsourced subjective knowledge acquisition, our target is to maximize the acquired knowledge under the crowdsourcing budget.
Next, we define the \textit{Crowdsourced Subjective Knowledge Acquisition} problem.

\begin{definition}[Crowdsourced Subjective Knowledge Acquisition]
Given a list of ST pairs, a knowledge base and a crowdsourcing budget $k$ (e.g. the number of crowdsourcing operation or monetary budget). The \textit{Crowdsourced Subjective Knowledge Acquisition} (\textit{CoSKA}) problem is to perform ST pair applying operations to acquire new subjective knowledge facts powered by crowdsourcing. The target is to maximize the number of derived knowledge facts under the given budget $k$.
\end{definition}
As described in Section~\ref{sec:introduction}, we propose a three stage approach for \textit{CoSKA}: ST pair selection, crowdsourced ST pair applying and knowledge inference. In next sections, we illustrate the details of each stage.



\section{ST Pair Selection}
\label{sec:pair_selection}
In this section, we describe the ST pair selection problem concerning the knowledge inference power. 
We first introduce the ST pair extraction method; then, we introduce the \textit{subjective resemble relationship} among ST pairs and define the knowledge inference rules based on the relationship. We then formulate the ST pair selection problem as a \textit{Maximum Knowledge Inference Problem}
 which is NP-hard and propose a diversity-aware forward greedy algorithm for ST pair selection.

\subsection{ST Pair Extraction}
\label{subsec:st_extraction}
As described, the input of our system \textit{CoSKA} is a set of ST pairs, and an ST pair consists of a subjective property which is usually an ``adjective'' and a type which is usually a ``noun phrase'' and corresponds to a type (class) entity in the KB. For example, a pair of $(big, city)$ is an ST pair as the $big$ is an adjective and $city$ can be mapped to a class entity in the knowledge base. In order to derive the commonly used ST pairs, we perform extraction from the news from New York Times.
We use three years' data which contains 
 167,958 news and 582,898,171 sentences. We process the data using NLP tools to identify adjective and noun phrase pairs. Similar with work~\cite{miningsubjective}, we use the synthetic patterns to extract information, i.e. ST pairs, from matched sentences. The pattern we adopted is shown in Figure~\ref{fig:pattern}. For example, given a sentence of ``\textit{Snakes are dangerous animals}'', an ST pair of $<dangerous,animal>$ are extracted and the ST pair of $<successful,film>$ can be extracted from sentence ``\textit{Titanic is the most successful film of all time}''.

\begin{figure}
	\centering
	\includegraphics[width=0.25\textwidth]{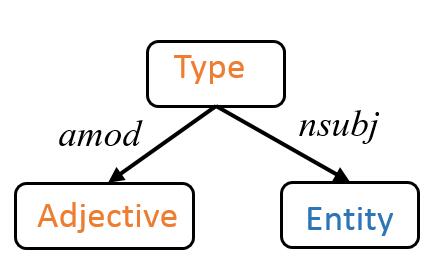}

	\caption{ST Pair Extraction Pattern.}
	\label{fig:pattern}
		\vspace{-2ex}
\end{figure}

After extraction using the pattern, we map the type from the extracted pairs to the given knowledge bases through textual similarity and filter out pairs that have no mapped class entity. 
In total, there are 40,582 mapped ST pairs with DBpedia.


\subsection{ST Pair Selection}
\label{subsec:grouping}
In order to reduce the number of ST pairs for crowdsourced subsequent ST pair applying and identify the most productive ST pairs in terms of the knowledge inference power. We define the \textit{Subjective Resemble Relationship} among ST pairs as follows:
\begin{definition}[ST Pair Subjective Resemble Relationship]
	\label{def:resemble}
	Given two ST pairs , $ST_1=(p_1^s,T_1)$, $ST_2=(p_2^s, T_2)$, a knowledge base $KB$ and an object $e$,  we define that $ST_1$ and $ST_2$ have the subjective resemble relationship on $e$, denoted as $ST_1 \approx_e ST_2$  if the following condition satisfies:
	\begin{itemize}
		\item Object $e$ is an instance of both types , i.e., $e\in I_{kb}(T_1)\land e\in I_{kb}(T_2)
		$, where $I_{kb}(T)$ denotes the instances of type $T$ in the $KB$.
		\item There exist a ``subclassOf'' relationship among two types in the $KB$, denoted as $<T_1,subclassOf,T_2>\in \mathcal{F}(KB) \lor$ $<T_2,subclassOf,T_1>\in \mathcal{F}(KB)$.
		\item If the two subjective properties are the same, synonymous or antonymous, denoted as $p_1^s\approx p_2^s$. If two subjective properties are synonymous or same, we have $p_1^s\approx^+ p_2^s$ and $p_1^s\approx^- p_2^s$ for antonymous.
	\end{itemize}
	Note that there are two kinds of subjective resemble relationships, $ST_1 \approx^+_e ST_2$ and $ST_1 \approx^-_e ST_2$ , we have $ST_1 \approx^+_e ST_2$ if $p_1^s\approx^+ p_2^s$ and $ST_1 \approx^-_e ST_2$ if $p_1^s\approx^- p_2^s$.
\end{definition}

If two ST pairs have the subjective resemble relationship on an entity, we can perform knowledge inference using the \textit{knowledge inference rule}:

\begin{lemma}[Knowledge Inference Rule]
	\label{lem:inference_rule}
	If we have a knowledge fact of $F=\{e,ST_1,l\}$ and two ST pairs where $ST_1=<p_1^s,T_1>$, $ST_2=<p_2^s,T_2>$:
	\begin{itemize}
		\item If $ST_1 \approx^+_e ST_2$, a new knowledge fact of $F'=\{e,ST_2,l\}$ can be inferred
		\item If $ST_1 \approx^-_e ST_2$, a new knowledge fact of $F'=\{e,ST_2,l'\}$ can be inferred, where $l'=\neg l$ 
	\end{itemize}
	
\end{lemma}

Next, we illustrate the \textit{ST Pair Subjective Resemble Relationship} and the \textit{Knowledge Inference Rule} through the following example.
\begin{example}
	\label{exa:inference}
	Given a knowledge base $KB$, four ST pairs, i.e., $ST_1=<old,Politician>, ST_2=<young,President>, ST_3=<big,City>, ST_4=<large, City>$ and four instances $e_1$=\{``Hillary Clinton''\}, $e_2$=\{``Barack Obama''\} , $e_3$=\{``New York''\} and $e_4$=\{"Los Angeles"\}. Referring to the knowledge in KB, we have that: $e_1$ is an instance of type ``Politician'', $e_2$ is an instance of type ``Politician'' and type ``President'', and $e_3, e_4$ are instances of the type ``City'', denoted as  $Type(e_1)$=\{``Politician''\}, $Type(e_2)$=\{``Politician'',``President''\}, $Type(e_3)$=$Type(e_4)$=\{``City''\}. Therefore, based on the Definition~\ref{def:resemble}, we can have the following ST pair subjective resemble relationships:
	\begin{itemize}
			\item[1).] $ST_1\approx^-_{e_2} ST_2$, as ``young'' and ``old'' are antonymous of each other, $e_2$ is an instance of both type ``President'' and ``Politician'' and ``President'' is a subclass of ``Politician'' ;
			\item[2).] $ST_3\approx^+_{e_3} ST_4$, as ``big'' and ``large'' are synonymous and the type of two ST pairs is ``City'' and $e_3$ is an instance of ``City''; 
			\item[3).] Similarly, we also have $ST_3\approx^+_{e_4} ST4$
	\end{itemize}

	Based on the ST pair subjective resemble relationships, we can infer new knowledge facts using the inference rule according to Lemma~\ref{lem:inference_rule}. If we have crowdsourced  subjective knowledge facts of $\mathcal{F}_{cr}=\{<e_2,ST_2,YES>,<e_3,ST_3,YES>\}$ (``Barack Obama is  a young president'' and ``New York is a big city''), we can get $F_{inf}=<e_2,ST_1,NO>$ (``Barack Obama is NOT an old politician'') and $F_{inf}=<e_3,ST_4,YES>$ (``New York is a large city'').
\end{example}

Based on the \textit{subjective resemble relationship} and the inference rule, we can construct a graph to model the ST pair subjective resemble relationship and the knowledge inference power.

\begin{definition}[ST Graph Model]
	\label{def:ST_Graph_Model}
	Given a knowledge base $KB=(E,L,R,P)$ and a set of ST pairs $\mathcal{P}=\{ST_1,ST_2,\cdots, ST_n\}$, we can construct a weighted graph $\mathcal{G}=\{\mathcal{V},\mathcal{W}\}$, where
	\begin{itemize}
		\item  Each vertex in $\mathcal{V}$ corresponds to an ST pair of $\mathcal{P}$.
		\item  There is an undirected edge $w_{ij}$ between node $v_i$ ($ST_i$) and $v_j$ ($ST_j$) if exists an instance $e\in E(KB)$ and $ST_i\approx_e ST_j$.
		\item  The weight of $w_{ij}$ is the number of entities on which two corresponding ST pairs have the subjective resemble relationship, denoted as $w_{ij}=|\mathcal{E}|$ where, $\forall e\in \mathcal{E}, ST_i\approx_e ST_j$.
	\end{itemize}
	  
\end{definition}  

For the example shown in Example~\ref{exa:inference}, we can have an ST graph with four vertices and two edges: $\mathcal{G}=\{\{v_1,v_2,v_3,v_4\}, \{w_{12},w_{34}\}\}$, where $v_i$ corresponds to $ST_i$ and $w_{12} =1, w_{34}=2$. 

Our target of ST pair selection is to identify the most beneficial ST pairs for subsequent crowdsourced ST pair applying to increase the acquired subjective knowledge facts.
In our work, based on the ST pair subjective resemble relationship and knowledge inference rules, we use the  number of  knowledge facts that can be inferred, i.e. the inference power, to measure the beneficial of selected ST pairs.
Therefore, we formulate the ST pair selection problem as a  \textit{Maximum Knowledge Inference Problem}.

\begin{definition}[Maximum Knowledge Inference Problem]
\label{def:st_selection}
Given a knowledge base $KB$ and a set of ST pairs, $\mathcal{P}=\{ST_1,ST_2,\cdots,ST_m\}$, the target is to select $k$ ST pairs to maximize the knowledge inference power.  
\end{definition}

Based on the \textit{ST Graph Model} defined in~\ref{def:ST_Graph_Model}, the \textit{maximum knowledge acquisition problem} is to select a set of nodes in the graph that maximize the total edge weight induced by the nodes. We can prove that the \textit{Maximum Knowledge Inference Problem} is NP-hard by a reduction 
\textit{Densest k-Subgraph problem}.
\begin{theorem}
	\label{the:proof_np}
	The Maximum Knowledge Inference Problem is NP-hard
\end{theorem}
\begin{proof}
Given an undirected graph $G=(V, E)$, the Densest k-Subgraph
(DkS) problem on $G$ is the problem of finding a subset $U\subseteq V $of vertices
of size $k$ with the maximum induced average degree. The average degree of the
subgraph will be denoted as
$2|E(U)|/k$. Here $|E(U)|$ denotes the number of edges in the subgraph induced by $U$. We construct the instance of ST graph as follows: Given $|V|$ ST pairs, each corresponds to a vertex in $G$, two ST pairs, $ST_i,ST_j$ have the subjective resemble relationship on a single object if there is an edge between to corresponding nodes, $v_i,v_j$ in $G$. Given the parameter $k$, the \textit{maximum knowledge inference problem} is to select $k$ nodes $S^*$ with maximum induced edge weight, $\sum_{e\in E(S^*)} W(e)=|E(S^*)|$. Therefore, under the same $k$, the optimal solution of the \textit{maximum knowledge acquisition problem} is equivalent to that of \textit{Densest k-Subgraph problem}
\end{proof}

The \textit{maximum knowledge acquisition problem} is NP-hard, a backward-greedy strategy, which repeatedly removes a vertex with the minimum weighted-degree in the remaining graph, until exactly k vertices are left, has an worst case approximation ratio of $[(\frac{1}{2}+\frac{n}{2k})^2-O(n^{-\frac{1}{3}}) , (\frac{1}{2}+\frac{n}{2k})^2+O(\frac{1}{n})]$ for $k$ in the range of $[\frac{n}{3},n]$ and $[2(\frac{n}{k}-1)-O(\frac{1}{k}) , 2(\frac{n}{k}-1)+O(\frac{n}{k^2})]$ for $k$ in the range of $[0,\frac{n}{3})$, where $n$ is the number of vertex~\cite{HkS}. However, the backward-greedy strategy is time consuming as it needs to iterate ($|V|-k$) times; moreover, the strategy does not consider the knowledge diversity, i.e. knowledge about different types, when making decisions, and therefore may results in top ST pairs share the same type. For example, in our experiment, we find that there are only two types from top 100 ST pairs by the backward-greedy strategy.  
In order to improve the efficiency and balance the subjective knowledge over various types, we propose a diversity-aware forward greedy strategy for ST pair selection: each time we select the pair with the maximum weight-degree, and add the pair to the result if the number of pairs with the same type does not exceed the given threshold.

The procedure of the diversity-aware forward greedy algorithm is illustrated in Algorithm~\ref{alg:Pair_Selection}. There are $k$ iterations (lines 2-7); in each iteration, we first pick the vertex with the maximum weight-degree (line 3); next, we check the number of vertices of the same type in the current result set, if the number dose not exceed the threshold, the vertex is added into the result set (lines 4-6). The complexity of Algorithm~\ref{alg:Pair_Selection} is $O(|\mathcal{V}|\cdot |\mathcal{W}|)$.
  
%

\begin{algorithm}[t]
\caption{Diversity-aware Forward Greedy Selection}
\label{alg:Pair_Selection}
\KwIn{ST Model Graph $\mathcal{G}=\{\mathcal{V},\mathcal{W}
\}$, Parameter $k$ and threshold $\delta$  }
\KwOut{A set of vertices $\mathcal{S}$}

$\mathcal{V} \leftarrow \emptyset$ \\
\While{$|\mathcal{V}|\leq k$}
{
		$v^*\leftarrow \arg\max_{v\in \mathcal{V}} WeightDegree(v)$ \\
		$T \leftarrow type(v^*)$\\
		\If{$Num(\mathcal{V},T)\leq \delta * k$}
		{
			$\mathcal{V} \leftarrow \mathcal{V}\cup v^*$\\
		}
		$\mathcal{V}\leftarrow \mathcal{V}\backslash v^*$\\
		
}
\Return $\mathcal{V}$
\end{algorithm}
%

\section{Crowdsourced ST Pair Applying}
\label{sec:crowdsourcing_knowledge}
\label{sec:st_apply}

For a given ST pair, the task of subjective knowledge acquisition is to identify the dominant opinion of whether the subjective property can be applied to the instances of the type, referred as \textit{crowdsourced ST pair applying}. However, asking the crowd for every instance is  too costly as a type in a KB could contain hundreds of thousands instances. Therefore,
we formulate the crowdsourced ST pair applying as a binary classification problem taking advantage of the knowledge in the KBs. For each instance, we decide whether the ST pair can be applied by the classification result. 

However, we do not have any labeled data for training the classifier. Therefore, we select a set of seed instances and ask the crowd to collect the corresponding subjective knowledge facts. We take the crowdsourced samples as the training data and train the classifier using the features extracted from the KB. As each type in a KB would have a set of properties/relations, we extract these properties/relations and list them in the crowdsourcing tasks, the crowd workers are also asked to mark which of the properties would affect the decision about whether the ST pair applies to the instances.  Then, we filter out the properties/relations with votes less than a threshold and training the classifier on the remaining properties. 

We adopt a
\textit{representative sampling} to sample the instances which explores the clustering structure of the large amount of unlabeled data and query the representative samples, i.e. samples from different clusters, as the training data. In our work, we cluster the instances using the knowledge from existing KBs and sample $k$ instances for each ST pair.

\section{Crowdsourcing Mechanism Design} 

\label{sec:crowd}
In this section, we first describe the human intelligent task (HIT) interface of \textit{CoSKA}, then we introduce the 
 answer aggregation strategy.
  
\textbf{HIT Interface}. 
Given an ST pair, we need to obtain the knowledge of \textit{whether a subjective property can be applied to the instances of the type}.
In order to reduce the cost, we design the task as a multiple choice question where each question contains 5 instances and the crowd worker is asked to select those that the given property can be applied. 
Furthermore, we list all the properties of the given type in each HIT and let the crowd worker to select the properties that would affect the decision. We compute the voting for each feature, and retain those with voting number exceeds the threshold (set through experimental studies) for further classification models. The HIT interface is shown in Figure~\ref{fig:HIT_subjective}. The illustrated HIT is for ST pair (big, City), we include five instances of the type ``City'' in each HIT, and ask the crowd to select the instances that has the attribute of ``big''. Also, there are properties related to instances of type ``City'' in the KB, e.g., \textit{Country}, \textit{areaLand}, \textit{foundingDate}, e.t.c, we list the properties and let the crowd workers to select relevant ones.  

\begin{figure}
\centering
\includegraphics[width=0.45\textwidth]{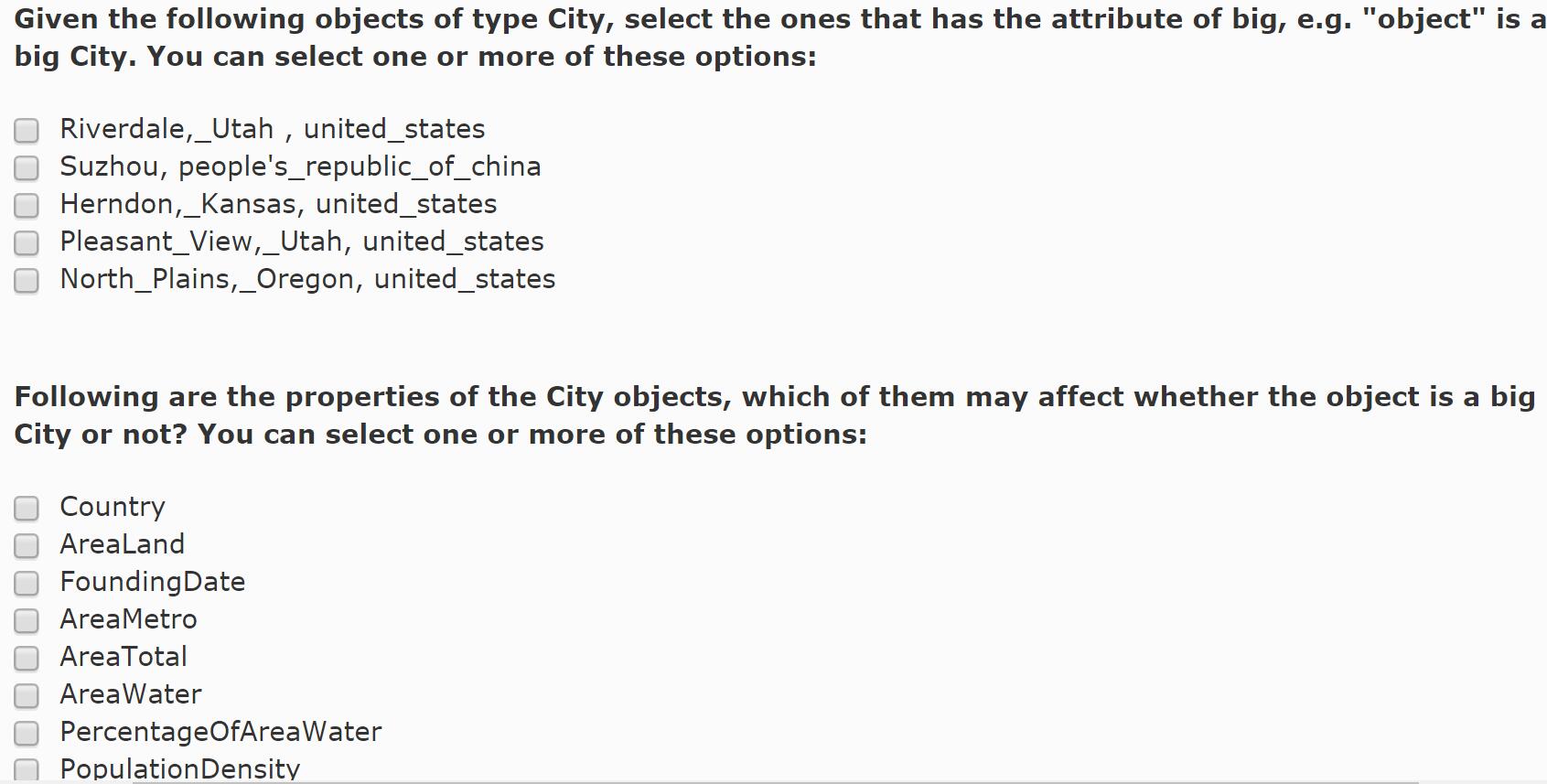}
\caption{Human Intelligent Task (HIT) interface.}
\label{fig:HIT_subjective}
\end{figure}

\textbf{Agreement-based Answer Aggregation}. After all the HITs are answered, for each selected instance sample, we can collect a set of \textit{yes} or \textit{no} answers of whether the given subjective property can be applied to it. We compute the degree of the agreement on each task: 
\begin{equation}\small
\label{equ:agreement}
\mathcal{A}(\mathcal{I},ST)=\frac{1}{|W|}\sum_{w_i\in W}  V(\mathcal{I},w_i)
\end{equation}

where $V(\mathcal{I}, w_i) = 1$ if the worker answer is \textit{yes} and 0 otherwise. The degree agreement of properties is computed in a similar way, and we retain the properties with agreement score at least $\theta_{P}$ as our subsequent classification features 

The \textit{agreement} score evaluates the confidence of the collected opinion among workers over the random answer. With the given threshold $\theta_A$ (which is set through experimental study, as the degree of agreement would vary with different ST pairs~\cite{miningsubjective}), we derive the \textit{dominant opinion}, denoted as ($\mathcal{DO}(\mathcal{I}, ST)$) of whether an ST pair $ST$ applies to the instance $\mathcal{I}$:

\begin{equation}\small
\label{equ:opinion}
\begin{split}
\mathcal{DO}(\mathcal{I}, p^s)\!=\!\begin{cases}
\text{\textit{yes}}   & \text{if} \quad \!\!\! \mathcal{A}(\mathcal{I},ST)-0.5\geq \theta_{A}\\
\text{\textit{no}}   & otherwise
\end{cases}
\end{split}
\vspace{-2ex}
\end{equation}

According to Equation~(\ref{equ:opinion}), we would obtain the positive opinion over the ST pair applies to an instance if the majority of the opinions is positive and this positive opinion has a high agreement (larger than $\theta_A$).
%
\section{Experiments}
\label{sec:experiment}
In this section, we evaluate \textit{CoSkA} on real knowledge base and crowdsourcing platform with extracted ST pairs. We describe the experimental setup in Section~\ref{subsec:experimental_setup}. Section~\ref{subsec:experimental_st} compares different methods for ST selection;  Section~\ref{subsec:experimental_crowdsourcing} verifies the proposed approaches for crowdsourced ST pair applying; Section~\ref{subsec:experiment_transfer} shows the test results of the proposed knowledge inference approach.

\subsection{Experimental Setup}
\label{subsec:experimental_setup}


\textbf{Knowledge Base}. We adopt DBpedia, which contains millions of knowledge facts (restricted to objective knowledge), classes (types) and instances as the KB in our experiments. The KB is represented as text files containing a list of triples of facts. The statistics DBpedia are given in Table~\ref{tab:statistics_kb}. The KB offers information for mapping extracted ST pair to the KB types and subjective knowledge inference.

\begin{table}
	\centering
	\caption{Statistics of  DBpedia}
	\label{tab:statistics_kb}
	\begin{tabular}{|l|l|l|l|l|} \hline
		\textbf{Dataset}  & \textbf{\#Facts} & \textbf{\#Entities} & \textbf{\#Classes} \\ \hline
	    DBpedia & 26,797,299 & 2,531,369 & 827 \\ \hline
	\end{tabular}
	\vspace{-2ex}
\end{table}

%
	
\textbf{Crowdsourcing Platform}. We use the real crowdsorucing platform, Amazon Mechanical Turk (AMT) as the platform to conduct the subjective knowledge acquisition tasks. As mentioned in Section~\ref{sec:crowd}, each HIT is designed as a multiple choice question, each question is assigned to 5 workers and
each worker would get a reward of \$0.02  for answering the task. In addition, we would pay for the AMT platform \$0.01 for each assignment. In our experimental settings, for each crowdsourced ST pair applying task, we would sample up to 200 instances.
 As illustrated in Figure~\ref{fig:HIT_subjective}, we include 5 instances in each HIT, therefore there are totally $\frac{200}{5}=40$ HITs which cost \$6. 
\subsection{ST Pair Selection}
\label{subsec:experimental_st}
As illustrated in the Section~\ref{sec:pair_selection}, we adopt a diversity-based forward greedy algorithm in our work for ST pair selection (Div-FGreedy). To evaluate the efficiency and the effectiveness of the propose algorithm, we use three metrics: 1). Induced Edge Weight of ST pairs, which indicates the inference power of selected ST pairs; 2). The number of different types of the ST pairs, which is used to evaluate the knowledge diversity; 3). Running time, which is recorded to demonstrate the efficiency of the algorithm. For comparison, we implement three other algorithms: backward greedy selection algorithm (BGreedy), forward greedy selection algorithm (FGreedy) and random selection algorithm (Random). We vary the number of selected ST pairs from 10$\sim$100, and fix the threshold for the diversity-based forward greedy algorithm (Div-FGreedy) to $0.1$ (the value of $\delta$ can be changed to satisfy the various diversity demand as the Div-FGreedy strategy can derive ST pairs with at least $\frac{1}{\theta}$ types), the results are shown in Figure~\ref{fig:exp_ST_Selection}.

 \begin{figure*}[t]
 	\centering
 	\subfigure[\small{Induced Edge Weight} ]{
 		\label{fig:exp_ST_Weight}
 		\includegraphics[width=0.27\textwidth,height=0.19\textheight]{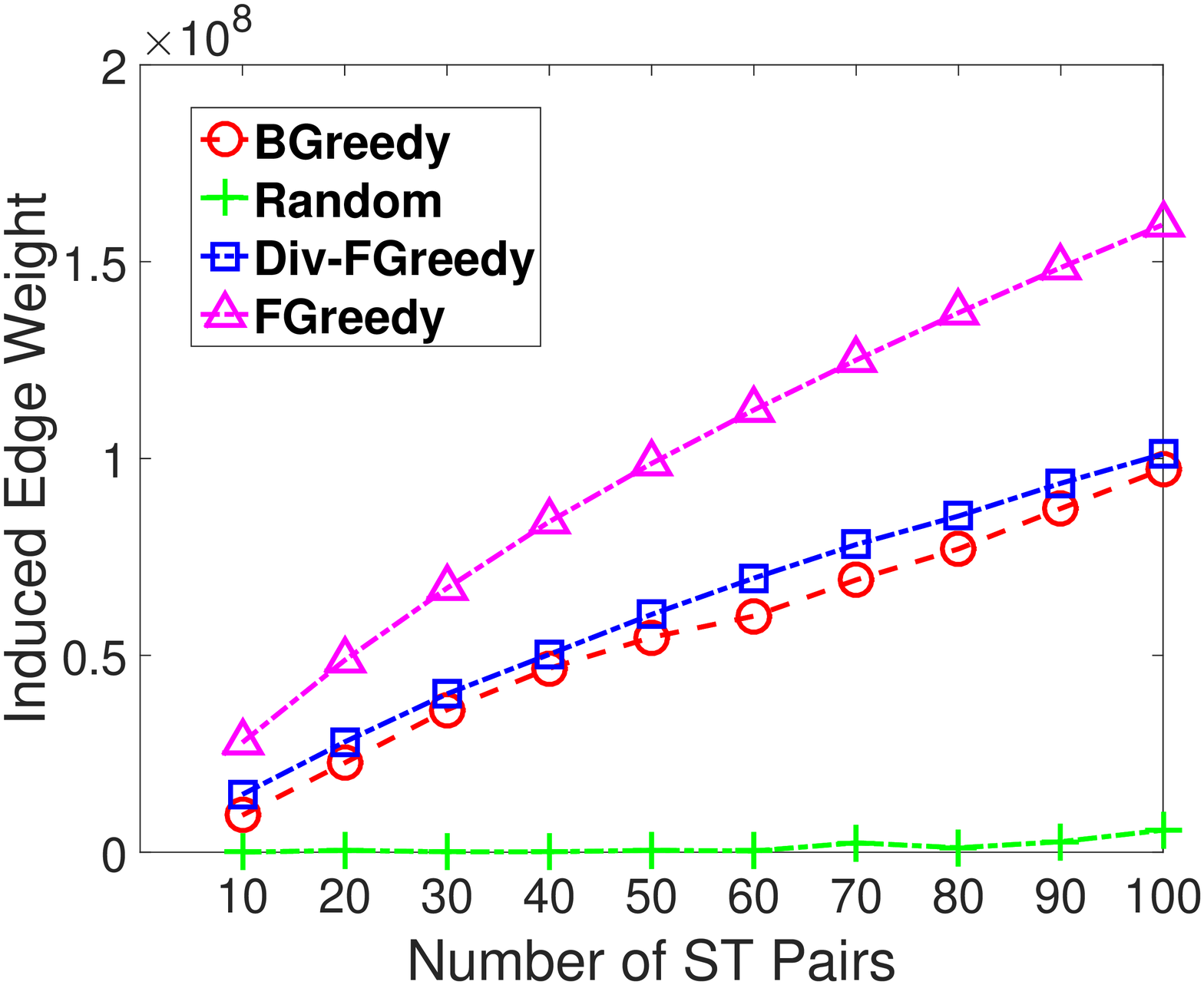}
 	}
 	\subfigure[\small{Number of Types}]{
 		\label{fig:exp_ST_Type}
 		\includegraphics[width=0.27\textwidth,height=0.19\textheight]{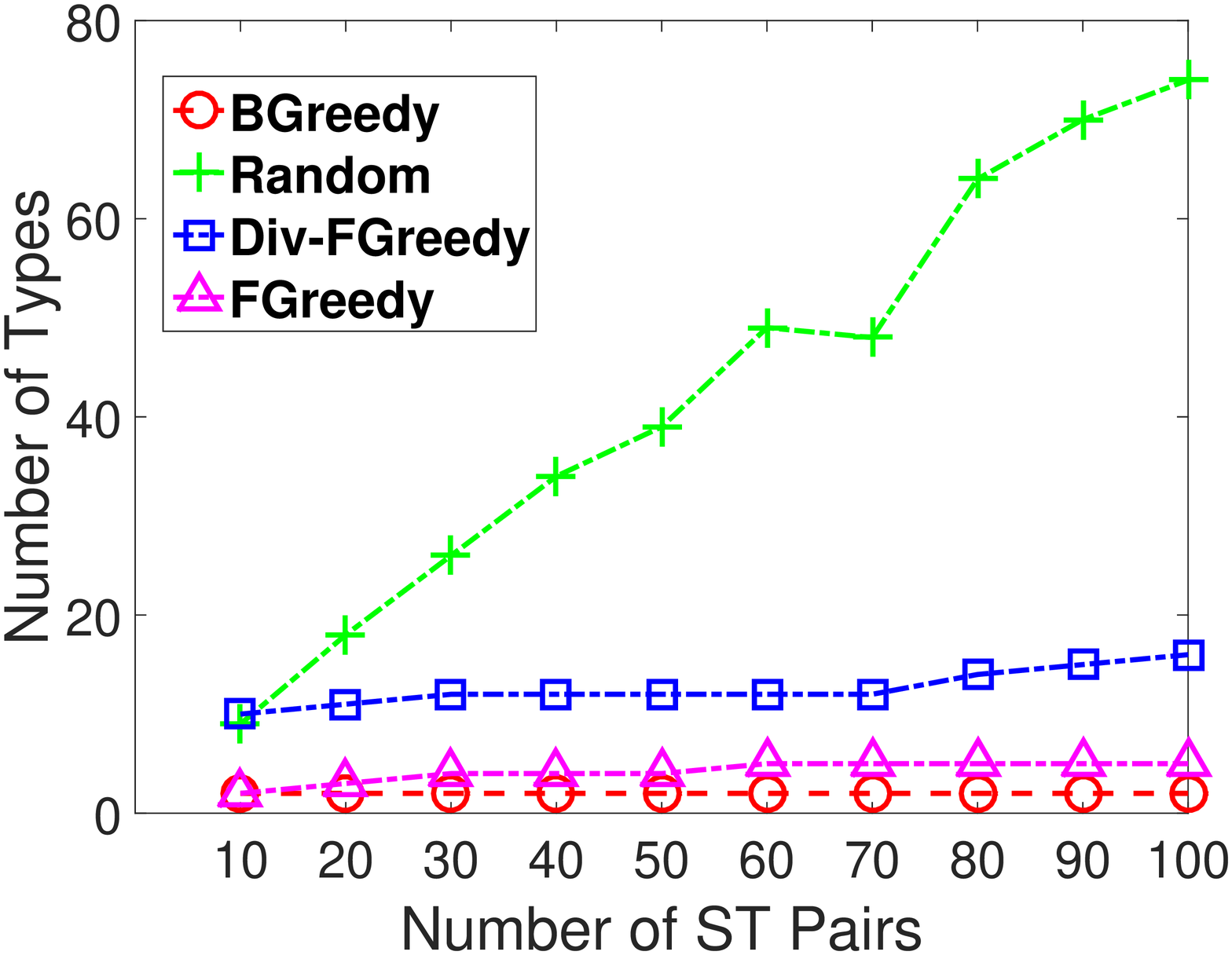}
 	}
 	\subfigure[\small{Running Time}]{
 		\label{fig:exp_ST_Time}
 		\includegraphics[width=0.27\textwidth,height=0.19\textheight]{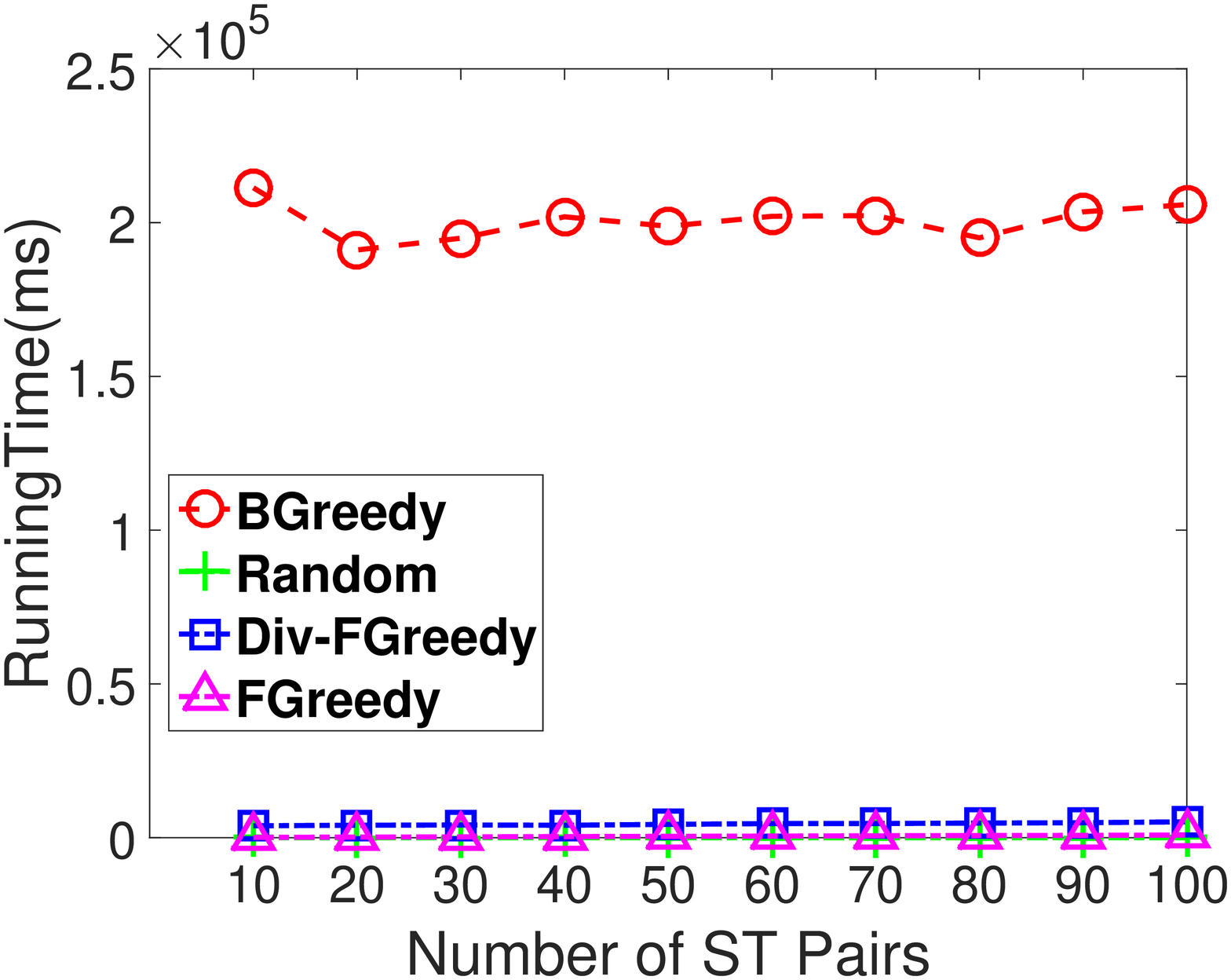}
 	}
 	 		\vspace{-2ex}
 	\caption{ST pair Selection Comparison.}
 	\label{fig:exp_ST_Selection}
 \end{figure*}

From Figure~\ref{fig:exp_ST_Weight}, we can find that the Random algorithm cannot achieve a good result, and the FGreedy and the Div-FGreedy algorithm outperform the BGreedy algorithm.
We can observe that in our experiments, the FGreedy strategy has the best performance in terms of the inference power (induced edge weight). However, from Figure~\ref{fig:exp_ST_Type}, we can see that the FGreedy algorithm would favor pairs with the same type as it does not consider the knowledge diversity, e.g. for the FGreedy algorithm, there would be only 5 types out of 100 selected ST pairs. The BGreedy strategy also has the same problem, e.g. the BGreedy only has 2 types out of 100 selected ST pairs. For the Div-FGreedy strategy, we have 16 types out of 100 selected ST pairs with the threshold set to $0.1$. We can see that the Random strategy can select pairs with larger type numbers as it selects ST pairs randomly. However, it does not consider the inference power when selecting ST pairs and thus results in ST pairs with quite low inference power as shown in Figure~\ref{fig:exp_ST_Weight}. 
For the running time shown in Figure~\ref{fig:exp_ST_Time}, we can observe that except for the BGreedy algorithm, all other three algorithms are quite efficient. To conclude, considering all three evaluation metrics, the Div-FGreedy algorithm can achieve a good inference power, guarantee the ST pair type diversity and is quite efficient.
\subsection{Crowdsourced ST Pair Applying}
\label{subsec:experimental_crowdsourcing}
For crowdsourced ST pair applying, we need to ask the crowd for a set of seed subjective knowledge facts and train the classifier based on the collected samples and features. In our experiments, we select 5 ST pairs through Div-FGreedy algorithm as test cases to evaluate the accuracy of our approach~\footnote{Note that the workflow of crowdsourced ST pair applying is same for each ST pair, to acquire more knowledge facts, we can perform crowdsourced ST pair applying for a larger number of ST pairs}. There are following configurations for the task: answer aggregation parameter $\theta_A$, feature selection parameter $\theta_P$ and classification models, the settings are illustrated in Table~\ref{tab:settings_CST}, where we mark our default settings in bold font. As we have no ground truth, we use 5-fold cross validation to test the performance of our approach.
 

\begin{table}[t]
	\centering
	\caption{Parameter Setting}
	\label{tab:settings_CST}
	\begin{tabular}{|l|l|} \hline
		\textbf{Factor}  & \textbf{Setting} \\ \hline
		$\theta_A$ & \textbf{0.1}, 0.3, 0.5\\ \hline
		$\theta_P$ & 0.1, 0.2, \textbf{0.3}, 0.4, 0.5\\ \hline
		Classifier & AdaBoost (AD), \textbf{Decision Tree (DT)}, \textbf{RBF-SVM},\\ & Nearest Neighbors (NN),  Random Forest (RF)\\ \hline
		ST Pairs & (big,City), (experienced, Athlete), (cute,Animal),\\& (old,Building),  (popular,Film)\\ \hline
	\end{tabular}
	\vspace{-2ex}
\end{table}
%

\textbf{Effect of Answer Aggregation Parameter}. There are two parameters in terms of the answer aggregation: $\theta_A$ for opinion aggregation and $\theta_P$ for feature selection. We first fix $\theta_A=0.1$ (in our settings, $\theta_A=0.1$ means at lease 60\% workers select the instance to have the given property), and vary the value of $\theta_P$ from 0.1$\sim$0.5 to compare the classification accuracy. The results are shown in Figure~\ref{fig:exp_CST_Property}~\footnote{Note that due to the space limit, we only show the results for four pairs and the result of (big,City) pair is summarized in Table~\ref{tab:exp_KI}}. From the results, we can observe that the classification accuracy would change with various $\theta_P$ value, the reason is that different $\theta_P$ have different filter power, i.e. with larger value of $\theta_P$, there would be less features remaining. Overall, the $\theta_P=0.3$ achieves best performance: from 0.1$\sim$0.3, there is an increasing trend of the accuracy for three pairs (old,Building), (experienced,Athlete) and (popular,film), and for pair (cute,Animal), the accuracy does not have much difference; for larger values (0.3$\sim$0.5), the accuracy would remain approximately the same. The reason is that the performance would change with different features (remained properties), and with lower value, less properties would be filtered therefore might retain those irrelevant properties as the training features.

Next, we check the effect of $\theta_A$ on the classification accuracy. We set the $\theta_P$ to 0.3 and vary the values of $\theta_A$ according to Table~\ref{tab:settings_CST}, the classification accuracy results are illustrated in Figure~\ref{fig:exp_CST_Opinion}. From the results we can observe that the classification models achieve best performance with the value of $\theta_A$ equals to 0.1, and would decrease as the value of $\theta_A$ increases. The reason is that with larger $\theta_A$, we have stronger restriction for deriving the dominant opinion of whether the property applies to an instance. For example, in our settings, when $\theta=0.5$, we would only obtain the opinion that the property applies to the instance if all the 5 workers gives the ``yes'' answer, which might results in missing some positive training samples and affect the classification performance. Overall, the classification models achieve best performance with the value of $\theta_A$ to 0.1. Therefore, we set the default value of $\theta_A$ to 0.1.

\textbf{Effect of Classification Model}. From the Figures~\ref{fig:exp_CST_Property} and ~\ref{fig:exp_CST_Opinion}, we can find that the performance of different classification models varies with different ST pairs. For pair (cute,Animal) different models have similar performance; for pairs (old,Building) and (experienced,Athlete) the RBF-SVM achieves the best performance whereas for (popular,film), the DT model outperforms other approaches.  We summarize the results of crowdsourced ST pair applying with the five pairs with all the parameters setting to the default value, the results are presented in Table~\ref{tab:exp_ST}. Overall, we can find that the RBF-SVM and DT model can achieve good performance for different ST pairs. 

To justify the effectiveness of our approach for subjective knowledge acquisition, we compare our results with the state-of-the-art technique, \textit{Surveyor}, proposed in~\cite{miningsubjective}. The reported accuracy of the approach by \textit{Surveyor} is 77\%. Compared our results with the \textit{Surveyor} approach, we can observe that except for the ST pair (cute,Animal), our approach can achieve better results, i.e. the pair (popular,Film) achieves accuracy of 79\% and all the other three pairs can achieve the accuracy of over 80\%. Therefore, our approach can perform accurate and scalable subjective knowledge acquisition with a low crowdsourcing budget (with up to 40 HITs and \$6 for each ST pair).

 \begin{figure*}[t]
 	\centering
 	\subfigure[\small{Pair (cute,Animal)} ]{
 		\label{fig:CST_Feature_Animal}
 		\includegraphics[width=0.23\textwidth,height=0.16\textheight]{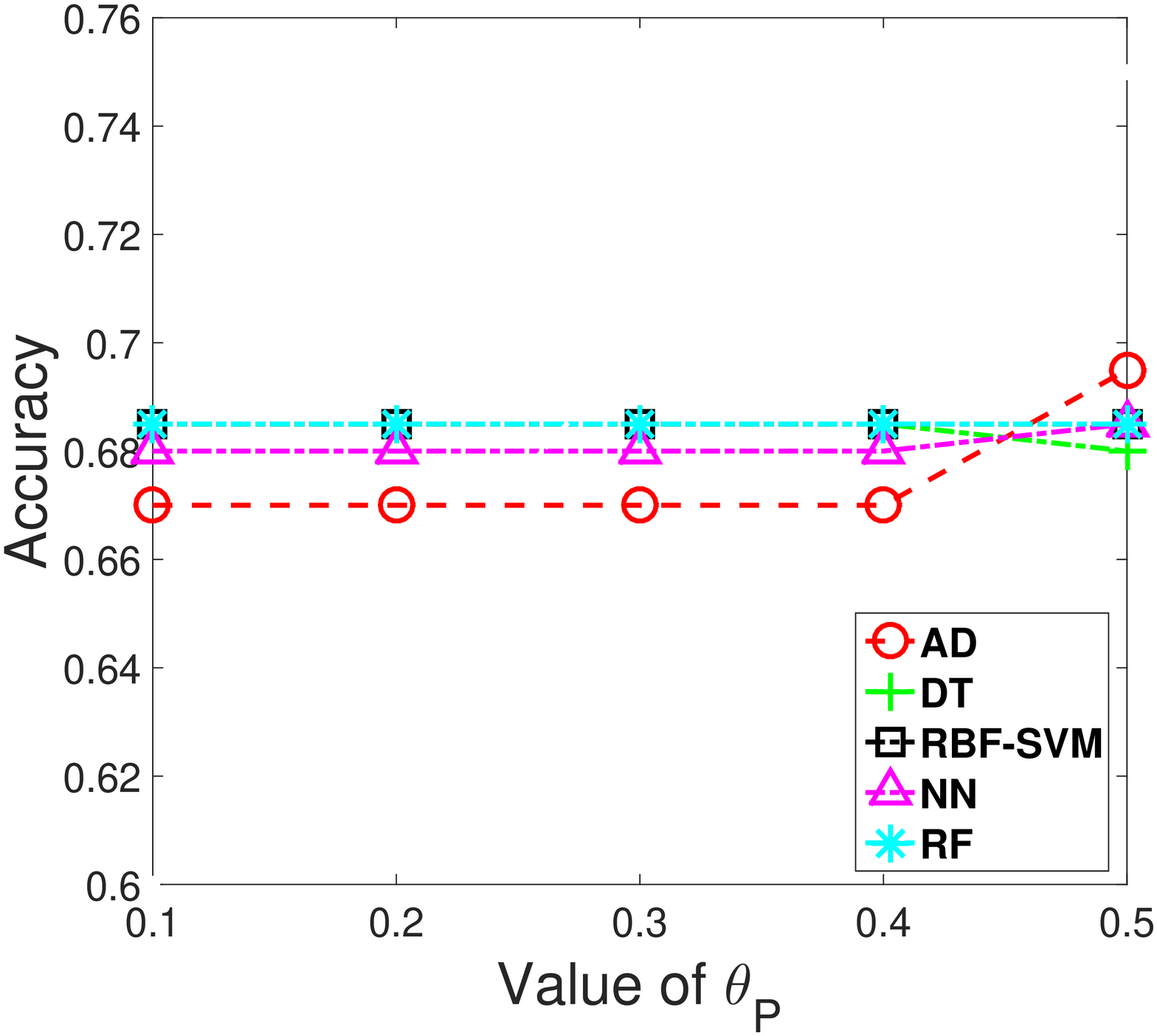}
 	}
 	\subfigure[\small{Pair (old,Building)}]{
 		\label{fig:CST_Feature_Building}
 		\includegraphics[width=0.23\textwidth,height=0.16\textheight]{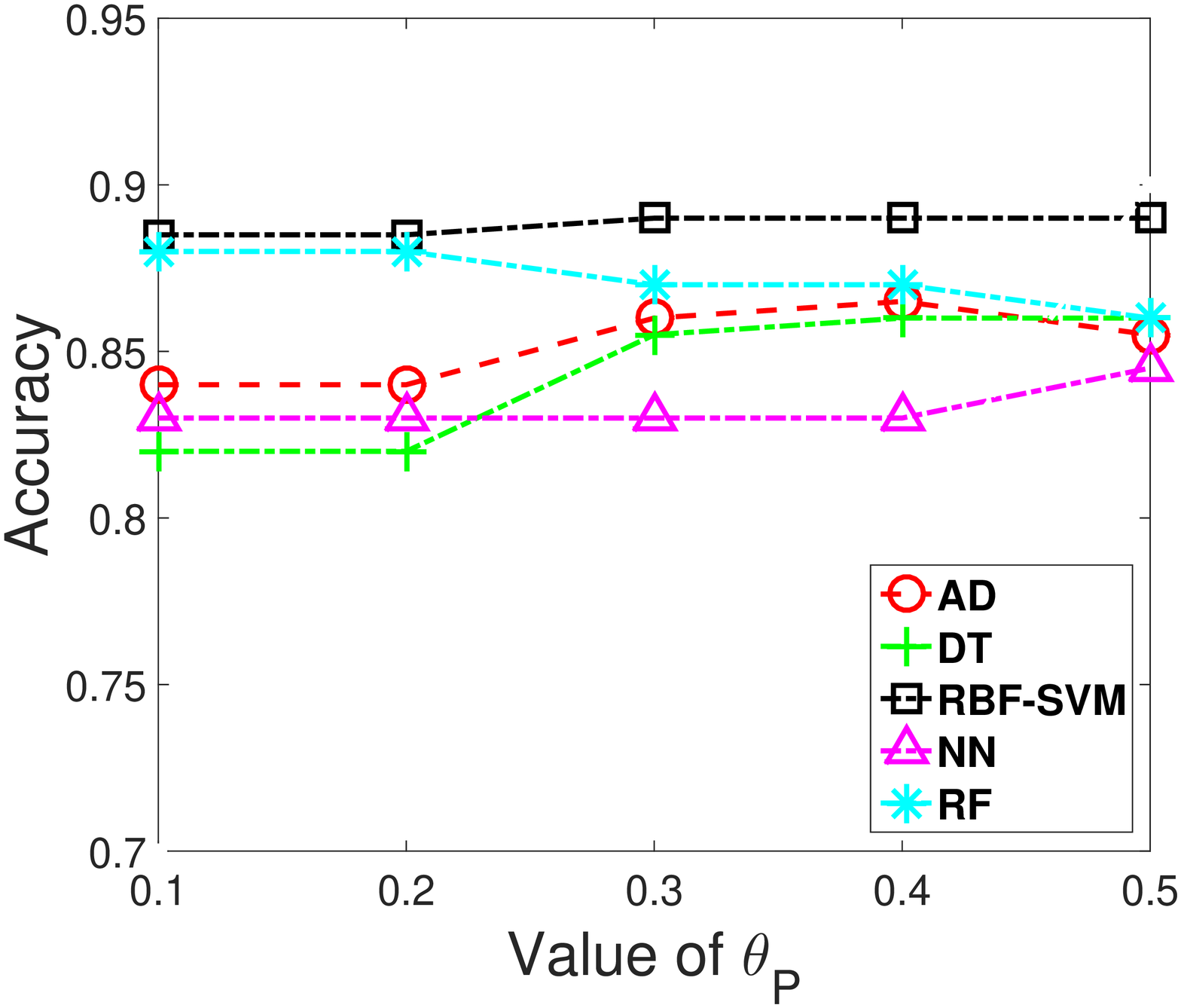}
 	}
 	\subfigure[\small{Pair (experienced,Athlete)}]{
 		\label{fig:CST_Feature_Athlete}
 		\includegraphics[width=0.23\textwidth,height=0.16\textheight]{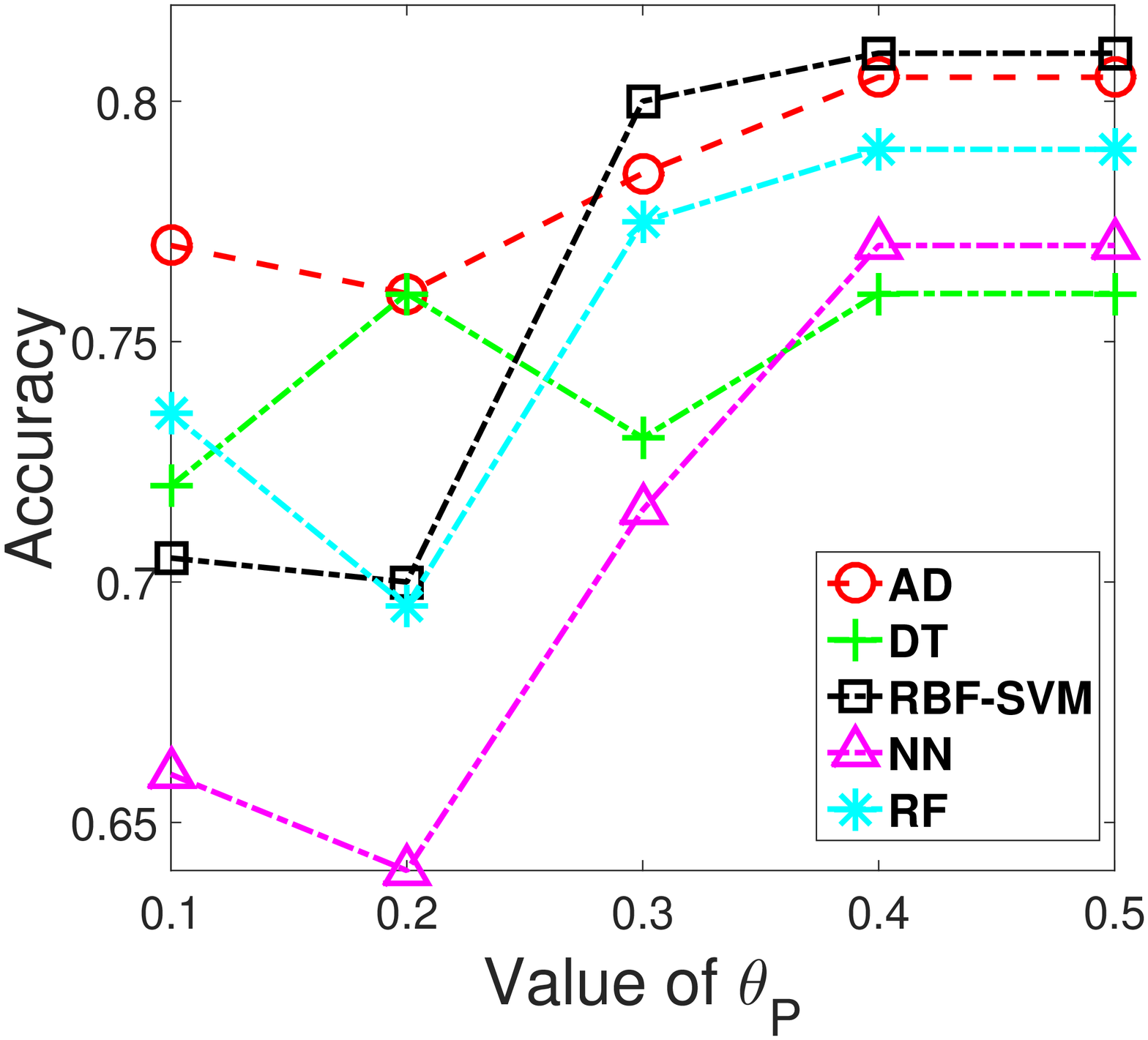}
 	}
 	\subfigure[\small{Pair (popular,Film)}]{
 		\label{fig:CST_Feature_Film}
 		\includegraphics[width=0.23\textwidth,height=0.16\textheight]{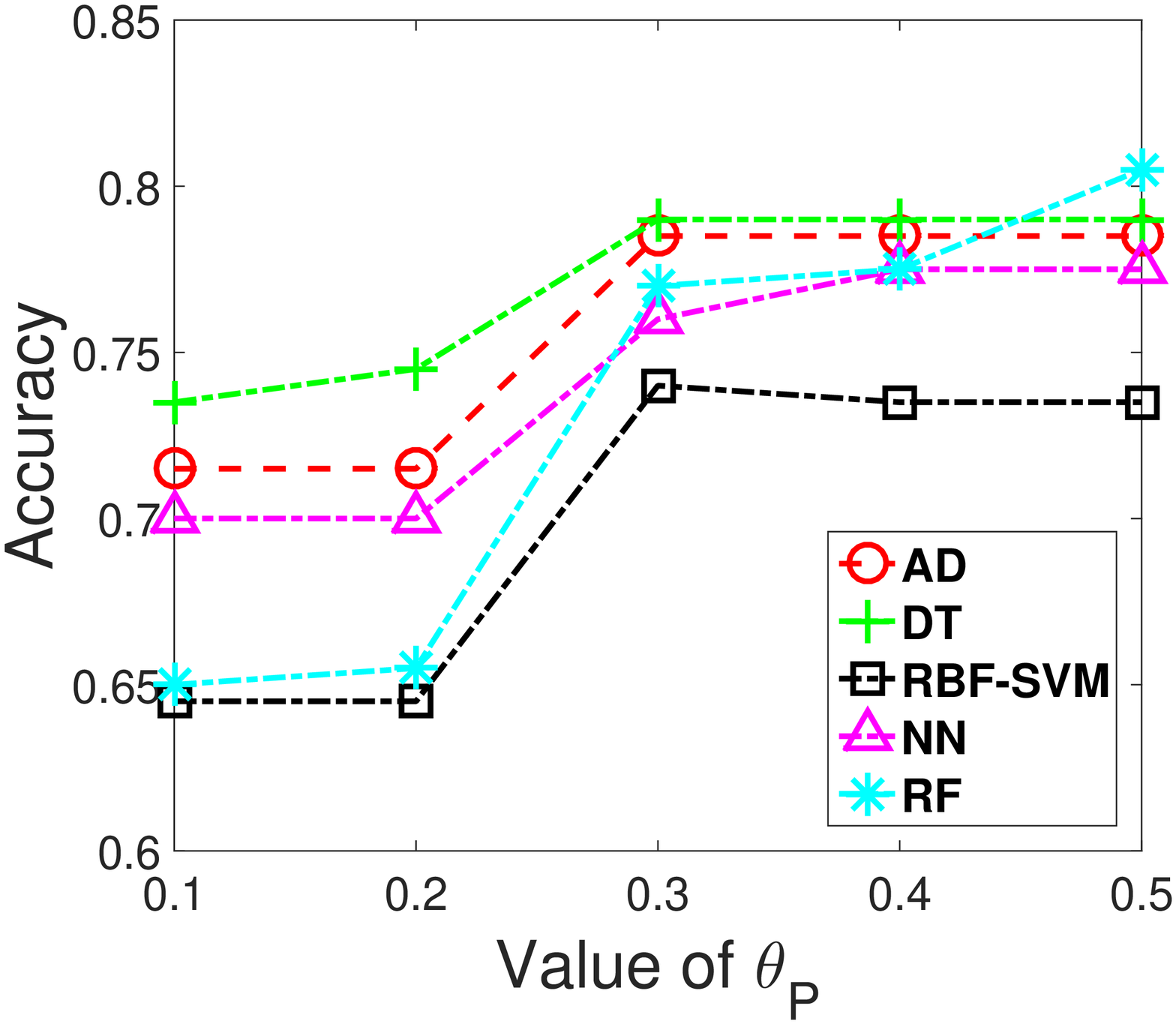}
 	}
 	 		\vspace{-2ex}
 	\caption{Results on varying Feature Selection Parameter ($\theta_P$). }
 	\label{fig:exp_CST_Property}
 \end{figure*}

\begin{figure*}[t]
	\centering
\subfigure[\small{Pair (cute,Animal)} ]{
	\label{fig:CST_Opinion_Anima1l1}
	\includegraphics[width=0.23\textwidth,height=0.16\textheight]{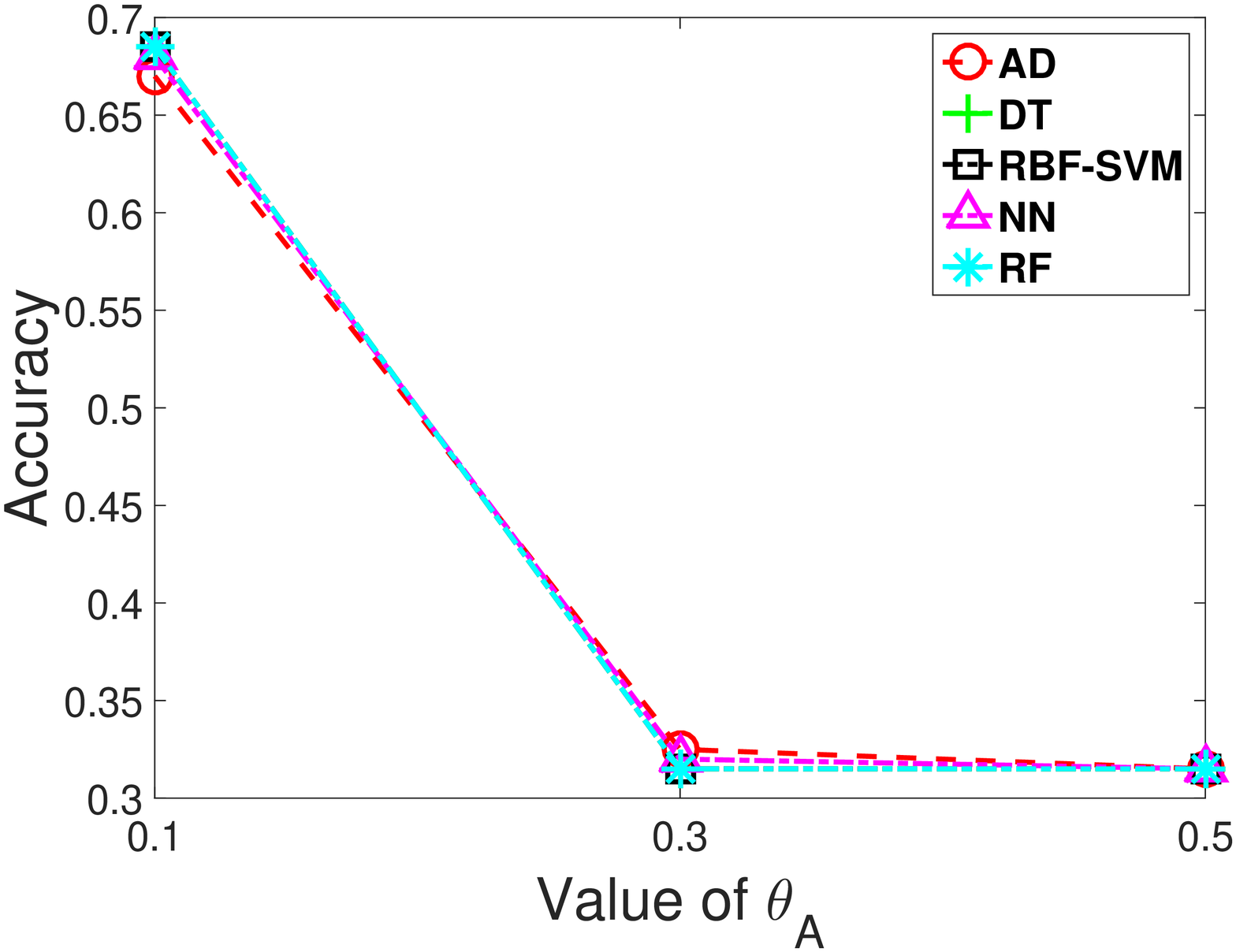}
}
\subfigure[\small{Pair (old,Building)}]{
	\label{fig:CST_Opinion_Buildin1g1}
	\includegraphics[width=0.23\textwidth,height=0.16\textheight]{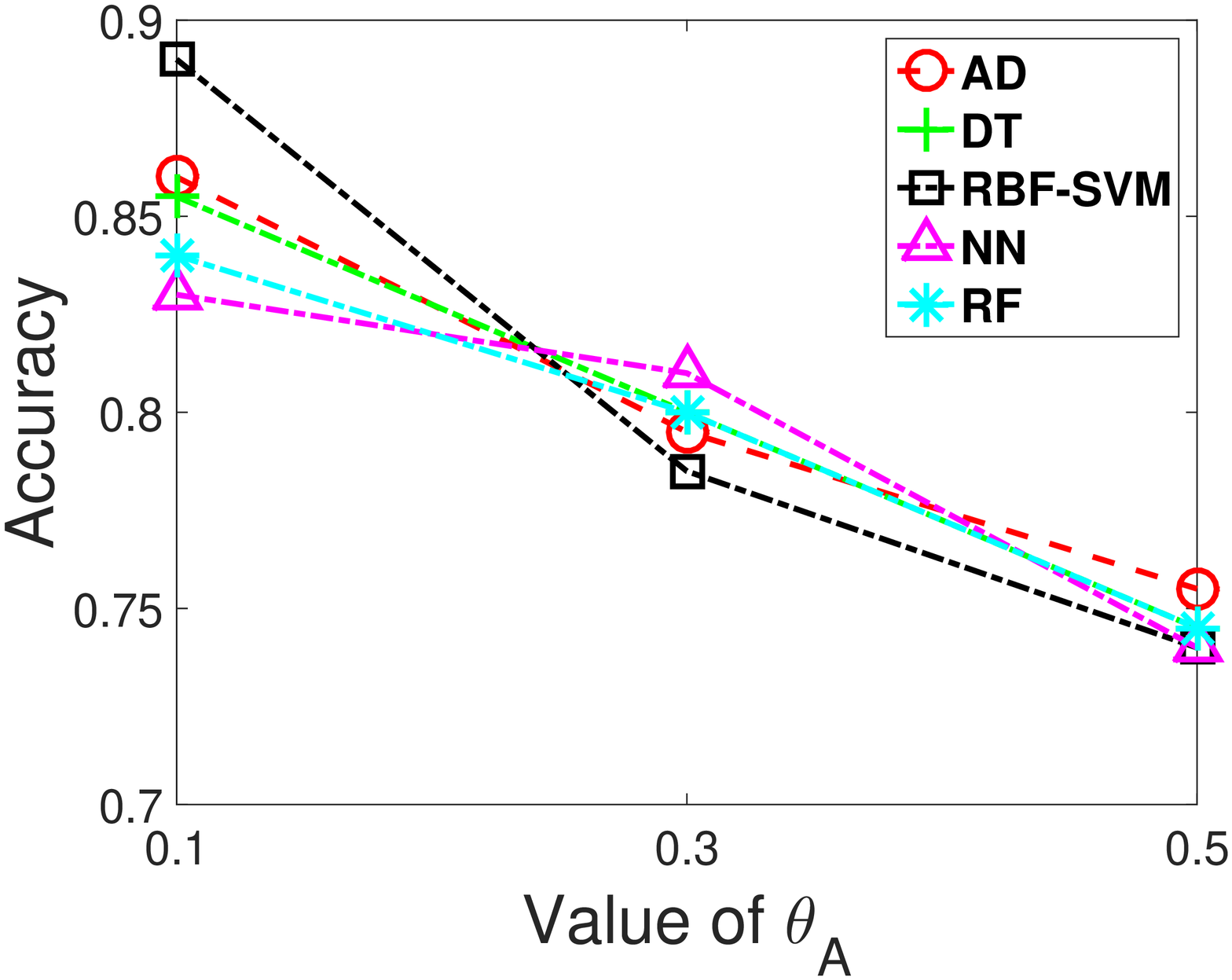}
}
\subfigure[\small{Pair (experienced,Athlete)}]{
	\label{fig:CST_Opinion_Athlete11}
	\includegraphics[width=0.23\textwidth,height=0.16\textheight]{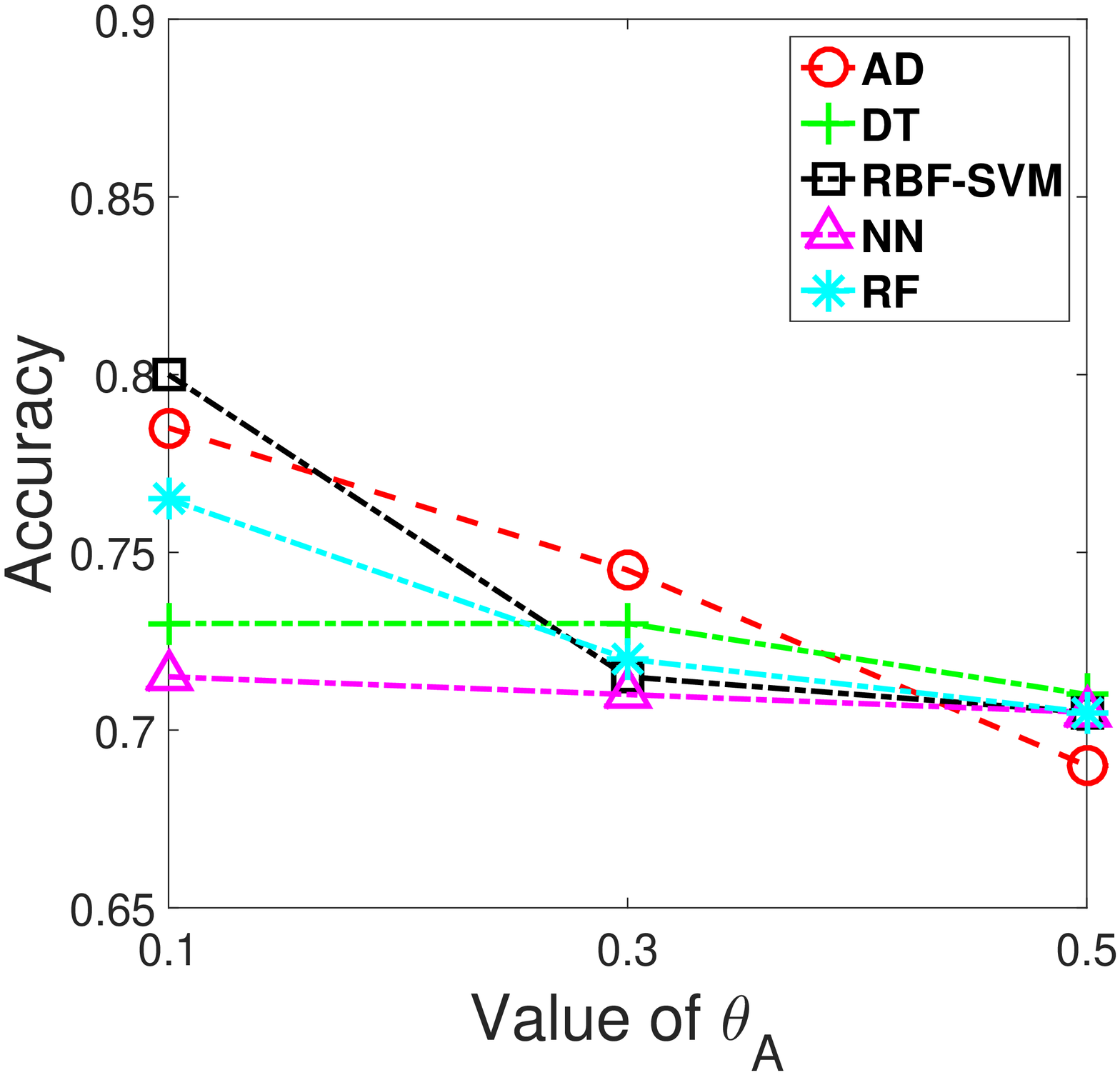}
}
\subfigure[\small{Pair (popular,Film)}]{
	\label{fig:CST_Opinion_Film11}
	\includegraphics[width=0.23\textwidth,height=0.16\textheight]{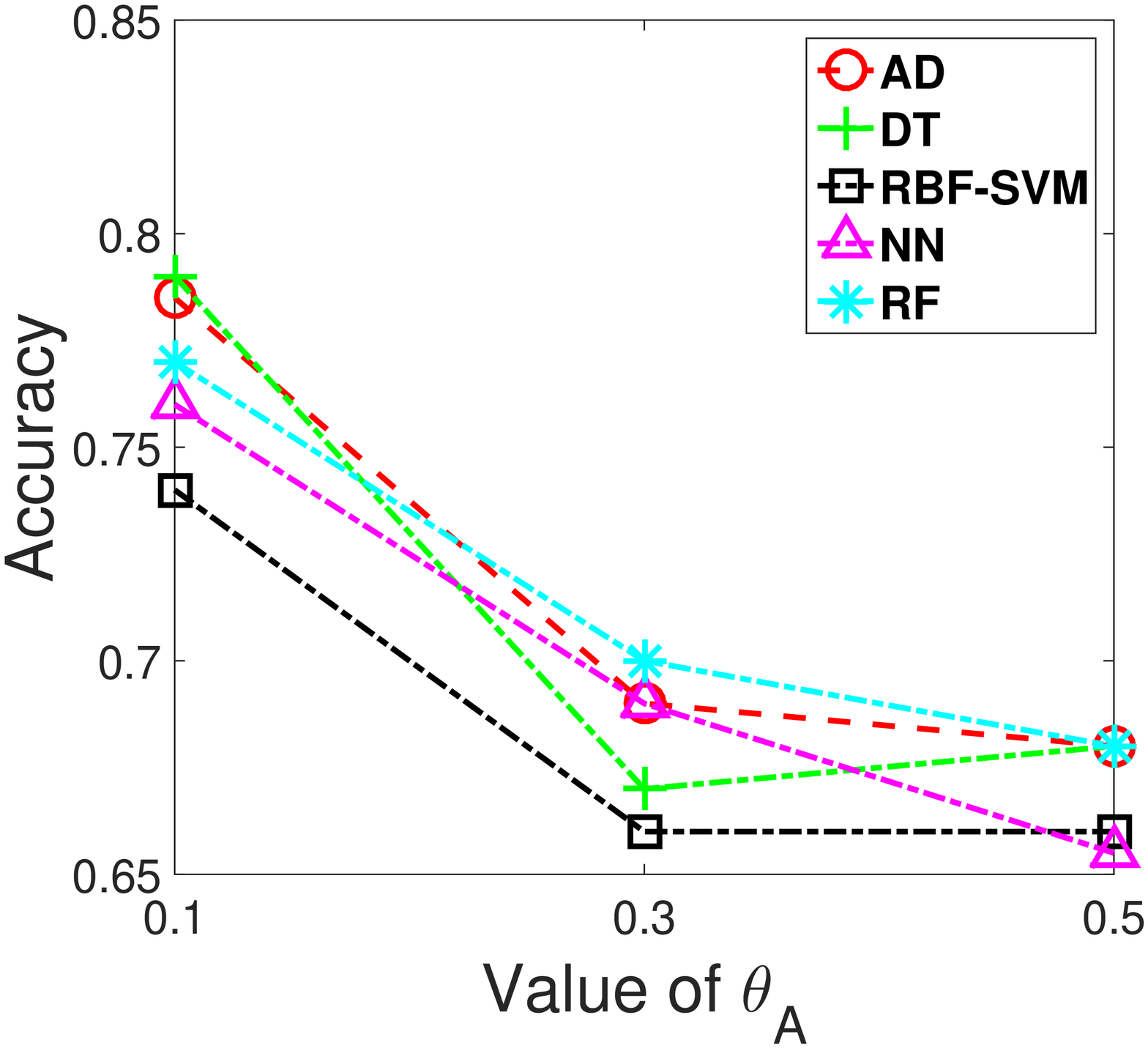}
}
	 		\vspace{-2ex}
	\caption{Results on varying Opinion Aggregation Parameter ($\theta_A$). }
	\label{fig:exp_CST_Opinion}
   
\end{figure*}

\begin{table}[t]
	\centering
	\vspace{-1ex}
	\caption{Crowdsourced ST Pair Applying Performance}
	\label{tab:exp_ST}
	\begin{tabular}{|c|c|c|} \hline
		\textbf{ST Pair}  & \textbf{Classification Model} & \textbf{Accuracy} \\ \hline
	(big,City) & DT  &  0.925 \\ \hline
	(experienced, Athlete) & RBF-SVM  & 0.80 \\ \hline
	(cute,Animal) & RBF-SVM/RF/DT  &  0.685 \\ \hline
	(old,Building) & RBF-SVM  & 0.89 \\ \hline
	(popular,Film) & DT  & 0.79\\ \hline
	\end{tabular}
\end{table}

\subsection{Knowledge Inference}
\label{subsec:experiment_transfer}
After the crowdsourced ST pair applying process, we have acquired a set of subjective knowledge facts, either collected from the crowd or obtained through the classification model. Then, we can perform knowledge inference to acquire more knowledge facts. Some resemble relationship of the selected ST pairs are shown in Table~\ref{tab:exp_KI_STpairs}. In order to verify the effectiveness of our knowledge inference approach, we evaluate two metrics: the number of inferred facts and the accuracy of inferred facts. As we do not have the ground truth, we sample 100 facts for each pair and verify the correctness manually. We ask three students to label whether the fact is correct or not and derive the answer by majority voting. The results of the knowledge inference performance are shown in Table~\ref{tab:exp_KI}.

From the results, we can observe that, on the one hand, large amount of knowledge facts could be inferred through the subjective knowledge inference approach; on the other hand, the inferred knowledge facts have high accuracy. Compare the accuracy results with those presented in Table~\ref{tab:exp_ST}, we can observe that the inferred knowledge of each ST pair has close but higher accuracy than the facts acquired by crowdsourced ST pair applying, which confirms that our proposed knowledge inference process does not introduce significant noises and verifies the high quality of our knowledge inference rules. To conclude, our knowledge inference approach can help to derive more high quality knowledge facts compared with that only using the crowd and classification models in the crowdsourced ST pair applying process. 

%
\begin{table}[t]\small
	\centering
	\caption{Subjective Resemble Relationship}
	\label{tab:exp_KI_STpairs}
	\begin{tabular}{|c|c|c|} \hline
		\textbf{ST Pair}  &\textbf{Subjective Resemble Relationship Pairs}   \\ \hline
		(big,City) &(small,City) , (big, Settlement) , (large,Place) \\  \hline
		(experienced, Athlete) & (experienced, SoccerPlayer) , (trained, Boxer) \\ \hline
		(cute,Animal) & (lovely, Animal) , (lovely,Species)   \\ \hline
		(old,Building) & (old, Hotel) , (new, Museum) \\ \hline
		(popular,Film) & (popular, Work) , (neglected, Film)  \\ \hline
	\end{tabular}
\end{table}


\begin{table}[t]\small
	\centering
	\caption{Knowledge Inference Performance}
	\label{tab:exp_KI}
	\begin{tabular}{|c|c|c|c|} \hline
		\textbf{ST Pair} &\textbf{\#Seed Facts} &\textbf{\#Inferred Facts}  &\textbf{Accuracy}   \\ \hline
		(big,City) & 10,354  & 93,186 & 0.92\\  \hline
		(experienced, Athlete) & 499 & 3,488 & 0.83 \\ \hline
		(cute,Animal) & 4,096 & 12,284   &0.76 \\ \hline
		(old,Building)& 233 &  1,398 &  0.93 \\ \hline
		(popular,Film)& 272 & 1,632 &  0.87 \\ \hline
	\end{tabular}
\end{table}

\section{Related Work}
\label{sec:related}
In this section, we discuss the  works related to  subjective knowledge acquisition, knowledge base enrichment and crowdsourcing.

Subjective knowledge acquisition is closely related to works that associating properties with entities. Some works have been conducted for commonsense knowledge acquisition~\cite{conceptnet}~\cite{resource}~\cite{contextual}~\cite{DBLP:conf/aaai/TandonMW14}. 
WebChild~\cite{webchild} presents a method for automatically constructing a large commonsense knowledge base, it contains triples that connect
nouns with adjectives via fine-grained relations. 
Entitytagger~\cite{entitytagger}, presented by Chakrabarti et al., automatically
associate descriptive phrases, referred to as etags (entity tags), to each entity. 
Instead of \textit{subjective properties}, these works focus on the less controversial and more objective properties, which is not related to obtaining \textit{dominant opinion}. The most similar work is \textsc{surveyor}, which mines the dominant opinion on the web content of whether a subjective property applies to a type. However, they does not consider to use the existing information in knowledge base and resorting to the crowd for subjective knowledge acquisition.
%

Knowledge base enrichment, completion and population have been widely studied. There are two mainstreams: internal methods, which use only the knowledge contained in the knowledge base to predict missing information~\cite{DBLP:conf/esws/VolkerN11}~\cite{AMIE}; external methods,which use sources of knowledge such as text corpora or other knowledge base to add new knowledge facts~\cite{DBLP:conf/www/WestGMSGL14}~\cite{DBLP:conf/www/NuzzoleseGPC12}~\cite{KBP_Successful}~\cite{KBP_TacklingChallenge}. However,  these works are limited to add \textit{objective knowledge} and neglect \textit{subjective knowledge}. Moreover, they do not consider to make use of a natural source of knowledge, the crowd, to complete/enrich the existing knowledge base.

Recently, 
the increasing popularity of crowdsourcing brings new trend to leverage the power of the crowd in knowledge acquisition, data integration and  many other applications. 
Kondreddi et al.~\cite{CrowdsourcedKnowledge} proposes a hybrid approach
that combines information extraction technique with human computation for knowledge acquisition. Marta et al.~\cite{CrowdsourcedKnowledgeAcquisition} presents a hybrid-genre workflow for games in crowdsourced knowledge acquisition process.  Works~\cite{DBLP:conf/chi/ChiltonLEWL13}~\cite{DBLP:conf/hcomp/BraggMW13}~\cite{DBLP:conf/icdm/MengT0C15} present approaches that use the wisdom of crowd to perform taxonomy construction.  Crowdsourcing also proved to have good performance in applications such as entity resolution~\cite{CrowdER}~\cite{DBLP:journals/pvldb/VesdapuntBD14}, schema matching~\cite{DBLP:journals/pvldb/ZhangCJC13}, translation~\cite{zaidan2011crowdsourcing} and so forth. 

\section{Conclusion}
\label{sec:conclusion}
In our work, we propose a  system
\textit{\underline{C}r\underline{o}wdsourced \underline{s}ubjective \underline{k}nowledge \underline{a}cquisition} (\textit{CoSKA}), 
for subjective knowledge acquisition powered by crowdsourcing and existing KBs.
The acquired knowledge can be encoded into existing KBs to perform KB enrichment in the subjective dimension which can bridge the gap between existing objective knowledge and the subjective queries.  Our \textit{CoSKA} system, consists of  three stages: ST pair selection, Crowdsourced ST pair applying and knowledge inference. To resolve the conflict between large scale knowledge facts and the limited crowdsourcing resource, we define subjective knowledge inference rules among ST pairs and perform knowledge inference to derive more knowledge facts. We formulate the ST pair selection problem as a \textit{Maximum Knowledge Inference Problem} which is NP-hard and we propose a diversity-aware forward greedy algorithm for ST pair selection. The crowdsourced ST pair applying problem is formulated as a classification task to further improve the system scalability. Experimental results on real knowledge base and crowdsourcing platform verify that our system, \textit{CoSKA},  could derive large amount accurate subjective knowledge facts with a comparative low crowdsourcing cost.     

\bibliographystyle{abbrv}
\bibliography{sigproc}  

\end{document}